\documentclass[11pt]{article}
\usepackage{subfigure}
\usepackage{multirow}
\usepackage{amsmath}
\usepackage{graphicx}
\usepackage[normalem]{ulem}
\usepackage{subfigure}
\usepackage{multirow}
\usepackage{algorithm}
\usepackage{algorithmic}
\usepackage{hyperref}
\usepackage{amsmath,amssymb}
\usepackage{bm}
\usepackage{enumerate}
\usepackage{color}
\usepackage{fullpage}
\usepackage{amsthm}

\def\math#1{$#1$}

\def\v#1{{\mathbf #1}}
\def\frac#1#2{{#1\over #2}}

\def\x{{\mathbf x}}

\def\z{{\mathbf z}}

\def\a{{\mathbf a}}
\def\b{{\mathbf b}}

\def\dotfil{\leaders\hbox to 1.5mm{.}\hfill}
\newcounter{rmnum}
\def\RN#1{\setcounter{rmnum}{#1}\uppercase\expandafter{\romannumeral\value{rmnum}}}
\def\rn#1{\setcounter{rmnum}{#1}\expandafter{\romannumeral\value{rmnum}}}

\newtheorem{definition}{Definition}
\newtheorem{theorem}{Theorem}

\newtheorem{lemma}[theorem]{Lemma}

\newcommand{\eps}{\varepsilon}
\newcommand{\ProbB}[1]{\ensuremath{\mathbb{P}\left[#1\right]}}
\newcommand{\OO}{\mathcal{O}}

\newcommand{\RR}{\mathbb{R}}

\newcommand{\abs}[1]{\ensuremath{\left|#1\right|}}

\newcommand{\trace}[1]{\ensuremath{\mathrm{\textbf{tr}}\left(#1\right)}}

\newcommand{\logdet}[1]{\ensuremath{\mathrm{logdet}\left(#1\right)}}

\newcommand{\qr}[1]{\ensuremath{\mathrm{qr}\left(#1\right)}}
\newcommand{\alogdet}[1]{\ensuremath{\widehat{\mathrm{logdet}} }\left(#1\right)}

\newcommand{\logm}[1]{\ensuremath{\mathrm{\textbf{\footnotesize log}}\left[#1\right]}}

\newcommand{\ignore}[1]{}
\newcommand{\Id}{\mathbf{I}}

\newcommand{\g}{\mathbf{g}}

\newcommand\code[1]{\texttt{#1}}
\newcommand\Cpp{{C\nolinebreak[4]\hspace{-.05em}\raisebox{.4ex}{\tiny\bf ++}}}
\usepackage{palatino}
\long\def\symbolfootnote[#1]#2{\begingroup%
\def\thefootnote{\fnsymbol{footnote}}\footnote[#1]{#2}\endgroup}

\newcommand{\FNormS}[1]{\mbox{}\|#1\|_\mathrm{F}^2}

\newcommand{\TNorm }[1]{\mbox{}\|#1\|_2  }

\newcommand{\transp}{^{\textsc{T}}}
\newcommand{\mat}[1]{{\ensuremath{\bm{\mathrm{#1}}}}}

\def\b{{\mathbf b}}

\def\v{{\mathbf v}}

\def\matA{\mat{A}}
\def\matB{\mat{B}}
\def\matC{\mat{C}}
\def\matD{\mat{D}}

\def\matI{\mat{I}}

\def\matL{\mat{L}}

\def\matQ{\mat{Q}}

\def\matU{\mat{U}}

\def\matX{\mat{X}}
\def\matY{\mat{Y}}

\DeclareMathSymbol{\Prob}{\mathbin}{AMSb}{"50}
\newcommand\remove[1]{}

\def\nnz{{ \rm nnz }}

\def\math#1{$#1$}

\def\frac#1#2{{#1\over #2}}

\DeclareMathSymbol{\R}{\mathbin}{AMSb}{"52}

\def\x{{\mathbf x}}

\def\z{{\mathbf z}}

\def\a{{\mathbf a}}
\def\b{{\mathbf b}}

\begin{document}

\title{A Randomized Algorithm for Approximating the Log Determinant of a
		Symmetric Positive Definite Matrix}
\author{Christos Boutsidis \thanks{Email: christos.boutsidis@gmail.com}
\and
Petros Drineas \thanks{Purdue University. West Lafayette, IN. Email: pdrineas@purdue.edu}
\and
Prabhanjan Kambadur \thanks{Bloomberg L.P. New York, NY. Email: pkambadur@bloomberg.net}
\and
Eugenia-Maria Kontopoulou \thanks{Purdue University. West Lafayette, IN. Email: ekontopo@purdue.edu}
\and
Anastasios Zouzias \thanks{Swisscom. Zurich, Switzerland. Email: anastasios.zouzias@swisscom.com}}
\date{}

\maketitle

\begin{abstract}
\noindent We introduce a novel algorithm for approximating the logarithm of the determinant of a symmetric positive definite (SPD) matrix. The algorithm is randomized and approximates the traces of a small number of matrix powers of a specially constructed matrix, using the method of Avron and Toledo~\cite{AT11}. From a theoretical perspective, we present additive and relative error bounds for our algorithm. Our additive error bound works for any SPD matrix, whereas our relative error bound works for SPD matrices whose eigenvalues lie in the interval $(\theta_1,1)$, with $0<\theta_1<1$; the latter setting was proposed in~\cite{icml2015_hana15}. From an empirical perspective, we demonstrate that a C++ implementation of our algorithm can approximate the logarithm of the determinant of large matrices very accurately in a matter of seconds.
\end{abstract}

\section{Introduction}

Given a matrix $\matA \in \R^{n \times n},$ the determinant of $\matA$, denoted by $\det(\matA)$, is one of the most important quantities associated with $\matA$. Since its invention by Cardano and Leibniz in the late 16th century, the determinant has been a fundamental mathematical concept with countless applications in numerical linear algebra and scientific computing. The advent of Big Data, which are often represented by matrices, increased the applicability of algorithms that compute, exactly or approximately, matrix determinants; see, for example,~\cite{leithead2005efficient,zhang2008log, zhang2007approximate,d2008first,hsieh2013big} for machine learning applications~(e.g., gaussian process regression) and~\cite{lesage2001spatial,kambadur2013parallel,friedman2008sparse,pace1997quick,pace2000method} for several data mining applications~(e.g., spatial-temporal time series analysis).

Formal definitions of the determinant include the well-known formulas derived by Leibniz and
Laplace; however, neither the Laplace nor the Leibniz  formula can be used to design an efficient, polynomial-time, algorithm to compute the determinant of $\matA$. To achieve this goal, one should rely on other properties of the determinant. For example, a standard approach would be to leverage the so-called $LU$ matrix decomposition or the Cholesky decomposition for symmetric positive definite matrices (SPD) to get an $O(n^3)$ deterministic algorithm to compute the determinant of $\matA$. (Recall that an SPD matrix is a symmetric matrix with strictly positive eigenvalues.)

In this paper, we are interested in approximating the logarithm of the determinant of a symmetric positive definite (SPD) matrix $\matA$. The logarithm of the determinant, instead of the determinant itself, is important in several settings~\cite{leithead2005efficient,zhang2008log, zhang2007approximate,d2008first,hsieh2013big,lesage2001spatial,kambadur2013parallel,friedman2008sparse,pace1997quick,pace2000method}.
\begin{definition}\textsc{[LogDet Problem definition]}\label{def:problem}
Given an SPD matrix $\matA \in \R^{n \times n}$, compute, exactly or approximately, $\logdet{\matA}$.
\end{definition}
\noindent Note that since all the eigenvalues of $\matA$ are strictly positive, the determinant of $\matA$ is strictly positive.
The best exact algorithm for the above problem simply computes the determinant of $\matA$ in cubic time and takes its logarithm. Few approximation algorithms have appeared in the literature, but they either lack a proper theoretical convergence analysis or do not work for all SPD matrices. We will discuss prior work in detail in Section~\ref{sec:related}.

\subsection{Our contributions}

We present a fast approximation algorithm for the problem of Definition~\ref{def:problem}. Our main algorithm (Algorithm~\ref{alg1a}) is randomized and its running time is
$$\OO\left(\nnz(\matA)\left(m \eps^{-2}+\log n \right)\log(1/\delta)\right),$$
where $\nnz(\matA)$ denotes the number of non-zero elements in $\matA$, $0<\delta<1$ denotes the failure probability of our algorithm, and (integer) $m >0$ and (real) $\varepsilon > 0$ are user-controlled accuracy parameters that are specified in the input of the algorithm.
The first step of our approximation algorithm uses the power method to compute an approximation to the dominant eigenvalue of $\matA$. This value will be used in a normalization~(preconditioning) step in order to compute a convergent matrix-Taylor expansion. The second step of our algorithm leverages a truncated matrix-Taylor expansion of a suitably constructed matrix in order to compute an approximation of the log determinant. This second step leverages a randomized trace estimation algorithm from~\cite{AT11}.

Let $\alogdet{\matA}$ be the value returned by our approximation algorithm~(Algorithm~\ref{alg1a}); let $\logdet{\matA}$ be the true log determinant of $\matA$; let $\lambda_{i}\left(\matA\right)$
denote the $i$-th eigenvalue of $\matA$ for all $i=1,\dots,n$ with
$\lambda_1(\matA)\ge\lambda_2(\matA)\ge\ldots\ge\lambda_n(\matA) > 0$; and let $\kappa\left(\matA\right) = \lambda_1(\matA)/\lambda_n(\matA)$ be the condition number of $\matA$. Our main result, proven in Lemma~\ref{thm1}, is that if
\begin{equation}\label{eqn:boundm}
m \ge \left\lceil 7\kappa\left(\matA \right)\log\left(\frac{1}{\varepsilon}\right)\right\rceil,
\end{equation}
then, with  probability at least $1-2\delta$,
\begin{equation}\label{eqn:mainresult}
\abs{ \alogdet{\matA} - \logdet{\matA} }  \leq
2\varepsilon \Gamma,
\end{equation}
where
$$
\Gamma = \sum_{i=1}^n  \log\left(7 \cdot \frac{\lambda_1(\matA)}{ \lambda_i(\matA) }\right).
$$
We now take a careful look at the above approximation bound. First, given our choice of $m$ in eqn.~(\ref{eqn:boundm}), the running time of the algorithm becomes
\begin{equation}\label{eqn:rtadd}
\OO\left(\nnz(\matA)\left(\kappa\left(\matA \right)\log\left(1/\varepsilon\right)\eps^{-2}+\log n \right)\log(1/\delta)\right).
\end{equation}
Thus, the running time of our algorithm increases linearly with the condition number of $\matA$.
The error of our algorithm scales with $\Gamma$, a quantity that is not immediately comparable to $\logdet{\matA}$. It is worth noting that the $\Gamma$ term increases \textit{logarithmically} with respect to the ratios $\lambda_1(\matA)/\lambda_i(\matA) \ge 1$. An obvious, but potentially loose upper bound for the sum of those ratios, is
\begin{equation}\label{eqn:gammabound}
\Gamma = \sum_{i=1}^n  \log\left( 7 \cdot \frac{\lambda_1(\matA)}{ \lambda_i(\matA) }\right)
\le n \cdot \log\left( 7 \kappa(\matA) \right).
\end{equation}
Our second result handles the family of SPD matrices whose eigenvalues all lie in the interval $(\theta_1,1)$, with $0<\theta_1<1$; this setting was proposed in~\cite{icml2015_hana15}.
In this case, a simplified version of Algorithm~\ref{alg1a} returns a relative error approximation to the log-determinant of the input matrix. Indeed, Lemma~\ref{thm2} proves that, with probability at least $1-\delta$,
\begin{equation*}
\abs{ \alogdet{\matA} - \logdet{\matA} }  \leq 2\varepsilon |\logdet{\matA}|.
\end{equation*}
The running time of the simplified algorithm is
\begin{equation}\label{eqn:rtrel}
\OO\left(\frac{\log(1/\varepsilon)\log(1/\delta)}{\varepsilon^2\theta_1} nnz(\matA)\right).
\end{equation}
Finally, we implemented our algorithm in \texttt{C++} and tested it on several large
dense and sparse matrices. Our dense implementation runs on top of
\texttt{Elemental}~\cite{poulson2013elemental}, a linear algebra library for
distributed matrix computations with dense matrices. Our sparse implementation
runs on top of \texttt{Eigen}~\footnote{\url{http://eigen.tuxfamily.org/}}, a
software library for sparse matrix computations. Our code is available to
download on Github (see Section~\ref{sxn:exps} for details and a link to our code).

\subsection{Related Work}\label{sec:related}

The most relevant result to ours is the work in~\cite{BP99}. Barry and Pace~\cite{BP99} described a randomized
algorithm for approximating the logarithm of the determinant of a matrix with special structure that we will describe below. They show that in order to approximate the logarithm of the determinant of a matrix
$\matA$, it suffices to approximate the traces of $\matD^{k}$, for $k=1,2,3...$ for a suitably constructed matrix $\matD$. Specifically, \cite{BP99} deals with approximations to SPD matrices $\matA$ of the form
$\matA = \matI_n - \alpha \matD,$
where $0 < \alpha < 1$ and all eigenvalues of $\matD$ are in the interval $\left[-1,1\right]$.
Given such a matrix $\matA$, the authors of~\cite{BP99} seek to derive an estimator $\alogdet{\matA}$ that is close to $\logdet{\matA}$. \cite{BP99} proved (using the so-called Martin expansion~\cite{Mar92}) that
$$ \log(\det(\matA)) = - \sum_{k=1}^{m} \frac{\alpha^k}{k} \trace{\matD^k} - \sum_{k=m}^{\infty} \frac{\alpha^k}{k} \trace{\matD^k}.$$
They considered the following estimator:
$$ \alogdet{\matA} =   \frac{1}{p}
\sum_{i=1}^{p}
\left(
\underbrace{
-n
\sum_{k=1}^{m}
\left(
  \frac{\alpha^k}{ k}   \frac{\z_i\transp \matD^k \z_i}{\z_i\transp \z_i}
\right)}_{V_i}
\right).
$$
All $V_i$ for $i=1\ldots p$ are random variables and the value of $p$ controls the variance of the estimator. The algorithm in~\cite{BP99} constructs vectors $\z_i$ $\in \R^n$ whose entries are independent identically distributed standard Gaussian random variables. The above estimator ignores the trailing terms
of the Martin expansion and only tries to approximate the first $m$ terms. \cite{BP99} presented the following approximation bound:
$$
\abs{ \alogdet{\matA} - \logdet{\matA} }  \leq \frac{n \cdot \alpha^{m-1}}{(m+1)(1-\alpha)}
 + 1.96 \cdot \sqrt{ \frac{\sigma^2}{p} },
$$
where $\sigma^2$ is the variance of the random variable $V_i$. The above bound fails with probability at most $0.05$.

We now compare the results in~\cite{BP99} with ours.
First, the idea of using the Martin expansion~\cite{Mar92} to relate the logarithm of the determinant
and traces of matrix powers is present in both approaches. Second, the algorithm of~\cite{BP99} is applicable to SPD matrices that have special structure, while our algorithm is applicable to any SPD matrix. Intuitively, we overcome this limitation of~\cite{BP99} by estimating the top eigenvalue of the matrix in the first step of our algorithm. Third, our error bound is much better that the error bound of~\cite{BP99}. To analyze our algorithm, we used the theory of randomized trace estimators of Avron and Toledo~\cite{AT11}, which relies on stronger measure-concentration inequalities than~\cite{BP99}, which uses the weaker Chebyshev's inequality.

A similar idea using Chebyshev polynomials appeared in the paper~\cite{pace2004chebyshev}; to the best of our understanding, there are no theoretical convergence properties of the proposed algorithm.  Applications to Gaussian process regression appeared in~\cite{leithead2005efficient,zhang2008log, zhang2007approximate}. The work of~\cite{Reusken2002} uses an approximate matrix inverse to compute the $n$-th root of the determinant of $\matA$ for large sparse SPD matrices. The error bounds in this work are a posteriori and thus not directly comparable to our bounds.

\cite{hunter2014computing} provides a strong worst-case theoretical result which is, however, only applicable to Symmetric Diagonally Dominant (SDD) matrices. The algorithm is randomized and guarantees that, with high probability,
$
\abs{ \alogdet{\matA} - \logdet{\matA} }  \leq \varepsilon \cdot n,
$
for a user specified error parameter $\varepsilon > 0$. This approach also uses the Martin expansion~\cite{Mar92} as well as ideas from preconditioning systems of linear equations with Laplacian matrices~\cite{spielman2004nearly}. The algorithm of~\cite{hunter2014computing} runs in time
$O\left(
\nnz(\matA) \varepsilon^{-2} \log^3\left(n\right)
\log^2\left( n \kappa(\matA)/\varepsilon \right)  \right).$
To compare to our approach, we need to combine the suboptimal upper bound for $\Gamma$ from eqn.~(\ref{eqn:gammabound}) with the bound of eqn.~(\ref{eqn:mainresult}). Then, we can run Algorithm~\ref{alg1a} with input
$$\varepsilon' = \frac{\varepsilon}{\log( 7 \kappa(\matA))},$$
instead of $\varepsilon$ to guarantee that the final error of our approximation will be bounded by $\varepsilon' n$. Then, we can observe that the running time of~\cite{hunter2014computing} depends
\textit{logarithmically} on the condition number of the input matrix $\matA$,
whereas our algorithm has a linear dependency on the condition number. Notice,
however, that our method is applicable to any SPD matrix while the method
in~\cite{hunter2014computing} is applicable only to SDD matrices; given current
state-of-the-art on Laplacian preconditioners it looks hard to extend the
approach of~\cite{hunter2014computing} to general SPD matrices.

Independently and in parallel with our work,~\cite{icml2015_hana15}
presented an algorithm using Stochastic Chebyshev Expansions for the
log-determinant problem. The algorithm is very similar in spirit to our
approach, using the Chebyshev instead of the Taylor expansion and achieves
relative-error guarantees for a special class of SPD matrices, namely matrices
whose eigenvalues all lie in the interval $(\theta_1,1-\theta_1)$ for some
$0<\theta_1<1/2$. As we already discussed, our algorithm also achieves a relative error bound under such an assumption; the only difference is that the running time of~\cite{icml2015_hana15} is proportional to $\sqrt{\frac{1}{\theta_1}}\log\frac{1}{\theta_1}$, whereas the running time of our approach (see eqn.~(\ref{eqn:rtrel})) is proportional to $\frac{1}{\theta_1}$. This slightly improved running time might be due to the use of the Stochastic Chebyshev Expansions. However, importantly, our algorithm works for \textit{any} SPD matrix, with arbitrary
eigenvalues. Not surprisingly, the added generality comes with a loss in accuracy and the relative error bound becomes an additive error bound.

Finally, two very recent papers~\cite{Han2016,Saibaba2016}\footnote{Both papers appeared after an earlier version of this paper was posted on ArXiv on March 2015 and cite this earlier version of our work.} presented algorithms to approximate the logdet of a matrix, highlighting the renewed importance of the topic. The work of~\cite{Saibaba2016} presents a very novel approach to approximate the logdet of a positive semi-definite matrix, using a randomized subspace iteration approach. To the best of our understanding, the relevant bounds in their work (Theorem 2 in~\cite{Saibaba2016}) are not directly comparable to our bounds. The work of~\cite{Han2016} follows the lines of~\cite{icml2015_hana15} and leverages the use of Chebyshev approximations to propose novel estimators for the trace of a matrix function. Among the many exciting applications of the proposed approach is an additive-error approach to approximate the logdet of any square non-singular matrix; the algorithm needs as inputs upper and lower bounds for all the singular values of the input matrix. Similar to the running time of our additive error algorithm in eqn.~(\ref{eqn:rtadd}), the time complexity of the proposed algorithm depends on the condition number of the input matrix (see Corollary 7 of~\cite{Han2016}).

We conclude by noting that common algorithms for the determinant computation assume floating point arithmetic and do not measure bit operations. If the computational cost is to be measured in bit operations, the situation is much more complicated and an exact computation of the determinant, even for integer matrices, is not trivial. We refer the interested reader to~\cite{eberly2000computing} for more details. 

\section{Preliminaries}

\subsection{Notation}

\noindent Let~\math{\matA,\matB,\ldots} denote matrices and let \math{\a,\b,\ldots} denote
column vectors. $\matI_{n}$ is the $n \times n$
identity matrix;  $\bm{0}_{m \times n}$ is the $m \times n$ matrix of zeros;
\trace{\matA} is the trace of a square matrix $\matA$;
the Frobenius and the spectral matrix-norms are:
$ \FNormS{\matA} = \sum_{i,j} \matA_{ij}^2$
and $\TNorm{\matA} = \max_{\TNorm{\x}=1}\TNorm{\matA\x}$.
We denote the determinant of a matrix $\matA$ by $\det(\matA)$ and the (natural) logarithm of the determinant of $\matA$ by $\logdet{\matA}$. We use $\log x$ to denote the natural logarithm of $x$ . Finally, given an event ${\cal E}$, $\ProbB{\cal E}$ denotes the probability of the event.

For an SPD matrix
$\matA \in \R^{n \times n},$ $\logm{\matA}$ is an $n \times n$ matrix defined as:
$
\logm{\matA} = \matU  \matD \matU \transp,
$
where $\matU  \in \R^{n \times n}$ contains the eigenvectors of $\matA$
and $\matD \in \R^{n \times n}$ is diagonal with entries being
$$\log(\lambda_1(\matA)), \log(\lambda_2(\matA)), \ldots{},
\log(\lambda_n(\matA)).$$
Let $x$ be a scalar variable that satisfies $|x|<1$. Then, using the Taylor expansion,
$$\log(1- x) = -\sum_{k=1}^{\infty} \frac{x^k}{k}.$$
A matrix-valued generalization of this identity is the following statement.
\begin{lemma}\label{lem:taylor}
Let $\matA \in \R^{n \times n}$ be a symmetric matrix whose eigenvalues all lie in the interval $(-1,1)$. Then,
\begin{equation*}\label{eq:matrixlog}
	\log(\matI_n -\matA) = -\sum_{k=1}^{\infty} \frac{\matA^k}{k}.
\end{equation*}
\end{lemma}

\subsection{Power method}\label{sec:power}
The first step in our algorithm for approximating the determinant of an SPD
matrix is to obtain an estimate for the largest eigenvalue of the matrix.  Given
an SPD matrix $\matA \in \R^{n \times n}$ we will use the
power-method (Algorithm~\ref{alg:power}) to obtain an accurate estimate of its largest eigenvalue. This
estimated eigenvalue is denoted by $\tilde{\lambda}_1(\matA)$.
\renewcommand\labelitemii{$\bullet$}
\begin{algorithm}[h!]
\begin{itemize}
\item Input: SPD matrix $\matA \in \R^{n \times n},$ integers $q,\ t > 0$
\item For $j=1,\dots,q$
\begin{enumerate}
\item Pick uniformly at random a vector $\x_0^j \in \{ +1, -1\}^{n} $
\item For $i=1,\dots,t$
\begin{itemize}
 \item $\x_i^j = \matA \cdot \x^j_{i-1}$
 \end{itemize}
\item Compute: $\tilde{\lambda}_1^j(\matA) =  \frac{{\x_t^j} \transp \matA \x_t^j}{{\x_t^j} \transp \x_t^j}$
 \end{enumerate}
\item Return: $\tilde{\lambda}_1(\matA) =  \max_{j=1\ldots q} \lambda_1^j$ (and the corresponding vector $\x_t=\x_t^j$)
\end{itemize}
\caption{Power method, repeated $q$ times.}\label{alg:power}
\end{algorithm}

\noindent Algorithm~\ref{alg:power} requires $\OO(qt (n + nnz(\matA)))$ arithmetic operations
to compute $\tilde{\lambda}_1(\matA)$. Lemma~\ref{lem:powerold} (see~\cite{LT_Lecture} for a proof)
argues that any $\tilde{\lambda}_1^j(\matA)$ is close to $\lambda_1(\matA)$.
\begin{lemma}\label{lem:powerold}
For any fixed $j=1\ldots q$, and for any $t > 0,$ $\eps > 0$, with probability at least $3/16$,
$$ \frac{(1 - \eps)}{1 + 4n(1-\eps)^{2t}}\lambda_1(\matA) \le
\frac{{\x_t^j} \transp \matA \x_t^j}{{\x_t^j} \transp \x_t^j}=\tilde{\lambda}_1^j(\matA).$$
\end{lemma}
\noindent Let $e = 2.718\ldots$ and let $\eps = 1-(1/e)$ and $t = \left\lceil\log \sqrt{4n}\right\rceil$; then, with probability at least $3/16$, for any fixed $j=1\ldots q$,
\begin{equation*}
\frac{1}{6}\lambda_1(\matA)\le\frac{1}{2e}\lambda_1(\matA) \le \tilde{\lambda}_1^j(\matA).
\end{equation*}
\noindent It is now easy to see that the largest value $\tilde{\lambda}_1(\matA)$ (and the corresponding vector $\x_t$) fails to satisfy the inequality $(1/6)\lambda_1(\matA) \le \tilde{\lambda}_1(\matA)$ with probability at most
$$\left(1-\frac{3}{16}\right)^q =\left(\frac{13}{16}\right)^q \leq \delta,$$
where the last inequality follows by setting $q = \left\lceil 4.82 \log(1/\delta)\right\rceil \geq \log(1/\delta)/\log(16/13)$. Finally, we note that, from the min-max principle, $\tilde{\lambda}_1(\matA) \leq \lambda_1(\matA)$. We summarize the above discussion in the following lemma.
\begin{lemma}\label{lem:power}
Let $\tilde{\lambda}_1(\matA)$ be the output of Algorithm~\ref{alg:power} with $q=\left\lceil 4.82 \log(1/\delta)\right\rceil$ and $t = \left\lceil\log \sqrt{4n}\right\rceil$. Then, with probability at least $1-\delta$,
$$\frac 1 6\lambda_1(\matA) \le \tilde{\lambda}_1(\matA) \leq \lambda_1(\matA).$$
The running time of Algorithm~\ref{alg:power} is $\OO\left(\left(n+nnz(\matA)\right)\log(n)\log\left(\frac{1}{\delta}\right)\right).$
\end{lemma}

\subsection{Trace estimation}\label{sxn:trace}

Even though computing the trace of a square $n \times n$ matrix requires only $O(n)$ arithmetic operations, the situation is more complicated when $\matA$ is given through a matrix function,
e.g., $\matA = \matX^2,$ for some matrix $\matX$ and the user only observes
$\matX$. For situations such as these, Avron and Toledo~\cite{AT11} analyzed several algorithms to estimate the trace of
$\matA$. Algorithm~\ref{alg:trace} and Lemma~\ref{thm:trace} present the relevant results from their paper.
\renewcommand\labelitemii{$\bullet$}
\begin{algorithm}[h!]
\begin{itemize}
\item Input: SPD matrix $\matA \in \R^{n \times n}$, accuracy $0 < \varepsilon < 1,$ and failure probability $0 < \delta < 1$.
\begin{enumerate}
\item Let $p = \left\lceil 20 \log(2/\delta) / \varepsilon^2 \right \rceil$
\item Let $\g_1,\g_2,\ldots, \g_p$ be a set of independent Gaussian vectors in $\RR^n$
\item Let $\gamma = 0$
\item For $i=1,\dots,p$
\begin{itemize}
 \item $\gamma = \gamma + \g_i^\top \matA \g_i$
 \end{itemize}
 \item $\gamma = \gamma / p$
 \end{enumerate}
\item Return: $\gamma$
\end{itemize}
\caption{Randomized Trace Estimation}\label{alg:trace}
\end{algorithm}
\begin{lemma}\label{thm:trace}
Let $\matA \in \R^{n \times n}$ be an SPD matrix, let $0<\eps < 1$ be an accuracy parameter, and let $0<\delta<1$ be a failure probability. If $\g_1,\g_2,\ldots, \g_p \in\RR^n$ are independent random standard Gaussian vectors, then, for $p= \left\lceil 20 \log(2/\delta) / \varepsilon^2 \right \rceil$, with probability at least $1 - \delta$,
\begin{equation*}\label{eq:trApprox}
\abs{
\trace{\matA} - \frac1{p} \sum_{i=1}^{p} \g_i^\top \matA \g_i
} \leq \eps \cdot \trace{\matA}.
\end{equation*}
\end{lemma}
\noindent The above lemma is immediate from Theorem~5.2 in~\cite{AT11}. 

\section{Additive error approximation for general SPD matrices}\label{sec:main}
Lemma~\ref{lem1} is the starting point of our main algorithm for approximating the determinant of a symmetric positive definite matrix. The message in the lemma is that
computing the log determinant of an SPD matrix $\matA$ reduces to the
task of computing the largest eigenvalue of $\matA$ and the trace of all the
powers of a matrix $\matC$ related to $\matA$.
\begin{lemma}\label{lem1}
Let $\matA \in \R^{n \times n}$ be an SPD matrix. For any $\alpha$ with
$\lambda_1(\matA) < \alpha,$
define
$\matB := \matA / \alpha$ and $\matC := \matI_n - \matB.$
Then,
\begin{equation*}
\logdet{\matA} = n \log(\alpha) -\sum_{k=1}^{\infty} \frac{\trace{\matC^k}}{k}.
\end{equation*}
\begin{proof}
Observe that $\matB$ is an SPD matrix with $\TNorm{\matB}<1$. It follows that
\begin{align*}
\logdet{\matA}
&= \log( \alpha^n \det(\matA /\alpha)) \\
&= n\log(\alpha) + \log \left(\prod_{i=1}^{n} \lambda_i(\matB)\right) \\
&= n\log(\alpha)  + \sum_{i=1}^{n}\log (\lambda_i(\matB)) \\
&= n\log(\alpha)  +  \trace{ \logm{\matB}}.
\end{align*}
Here, we used standard properties of the determinant, standard properties of the logarithm function, and the fact that (recall that $\matB$ is an SPD matrix),
$$\trace{\logm{\matB}} = \sum_{i=1}^{n} \lambda_i(\logm{\matB}) = \sum_{i=1}^{n} \log(\lambda_i(\matB)).$$
\noindent Now,
\begin{equation}\label{eq:trlog}
	\trace{ \logm{\matB}}
	= \trace{ \logm{\Id_n - (\Id_n - \matB)}}
	= \trace{ -\sum_{k=1}^{\infty}  \frac{\left(\Id_n - \matB\right)^k}{k}}
	 = -\sum_{k=1}^{\infty} \frac{\trace{ \matC^k}}{k}.
\end{equation}
The second equality follows by the Taylor expansion because all the
eigenvalues of $\matC = \Id_n - \matB$ are contained\footnote{Indeed,
$\lambda_i(\matC) = 1 - \lambda_i(\matB)$ and $0<\lambda_i(\matB) < 1 $ for all
$i=1\ldots n$.} in $(0,1)$ and the last equality follows by the linearity of the
trace operator.
\end{proof}

\end{lemma}

\subsection{Algorithm}
Lemma~\ref{lem1} indicates the following high-level
procedure for computing the logdet of an SPD matrix $\matA$:
\begin{enumerate}
\item Compute some $\alpha$ with $\lambda_1(\matA) < \alpha$.
\item Compute $\matC =  \matI_n - \matA / \alpha$.
\item Compute the trace of \emph{all} the powers of $\matC$.
\end{enumerate}

\noindent To implement the first step in this procedure we use the power iteration from the numerical linear algebra literature~(see Section~\ref{sec:power}). The second step is straightforward. To implement the third step,
we keep a finite number of summands in the expansion $\sum_{k=1}^{\infty} \trace{ \matC^k}$.
This step is important since the quality of the approximation, both theoretically and empirically, depends on the number of summands (denoted with $m$) that will be kept. On the other hand, the running time of the algorithm increases with $m$.  Finally, to estimate the traces of the powers of $\matC$, we use the randomized algorithm of Section~\ref{sxn:trace}. Our approach is described in detail in Algorithm~\ref{alg1a}; notice that step $7$ in Algorithm~\ref{alg1a} is an efficient way of
computing
$$
	\alogdet{\matA} := n\log(\alpha)- \sum_{k=1}^{m}
	\left(
	\frac{1}{p} \sum_{i=1}^{p}    \g_i^\top \matC^k \g_i
	\right)/k.
$$
\begin{algorithm}[h!]
\begin{algorithmic}[1]
\STATE {\bf{INPUT}}: $\matA \in \R^{n \times n}$, accuracy parameter $\eps > 0$, and integer $m >0$.
\STATE Compute $\tilde{\lambda}_1(\matA)$ using Algorithm~\ref{alg:power} with (integers)$t = \OO\left(\log n\right)$ and $q = \OO\left(\log(1/\delta)\right)$
\STATE Pick $\alpha = 7 \tilde{\lambda}_1(\matA)$
\STATE Set $\matC = \matI_n - \matA / \alpha$
\STATE Set $p = \left\lceil 20 \log(2/\delta) / \varepsilon^2 \right \rceil$
\STATE Let $\g_1,\g_2,\ldots, \g_p \in \RR^n$ be i.i.d. random Gaussian vectors.
\STATE For {$i=1,2\ldots, p$}
           \begin{itemize}
           \item  $\v_1^{(i)} = \matC \g_i$ and $\gamma_1^{(i)} = \g_i^\top \v_1^{(i)}$
           \item For {$k=2,\ldots , m$}
           \begin{enumerate}
           \item $\v_k^{(i)} : = \matC \v_{k-1}^{(i)}$.
           \item $\gamma_k^{(i)} = \g_i^\top \v_k^{(i)}$\
                    \text{(Inductively $\gamma_{k}^{(i)} = \g_i^\top \matC^k \g_i$)}
           \end{enumerate}
           \item EndFor
           \end{itemize}
\STATE EndFor
\STATE{\bf{OUTPUT}:
$\alogdet{\matA} = n\log(\alpha) - \sum_{k=1}^{m} \left(\frac1{p}\sum_{i=1}^{p} \gamma_k^{(i)} \right) / k$
}
\end{algorithmic}
\caption{\small{Randomized Log Determinant Estimation}}
\label{alg1a}
\end{algorithm}
\subsection{Error bound}
The following lemma proves that Algorithm~\ref{alg1a} returns an accurate approximation to the logdet of $\matA$.
\begin{lemma}\label{thm1}
Let $\alogdet{\matA}$ be the output of Algorithm~\ref{alg1a} on inputs $\matA,$ $m,$ and $\varepsilon$.
Then, with probability at least $1-2\delta$,
$$
\abs{ \alogdet{\matA} - \logdet{\matA} }  \leq
\left(\epsilon+\left(1-\frac{1}{7\kappa\left(\matA \right)} \right)^m\right) \cdot \Gamma,
$$
where
$
\Gamma = \sum_{i=1}^n  \log\left( 7 \cdot \frac{\lambda_1(\matA)}{ \lambda_i(\matA) }\right).
$
If $m \ge \left\lceil7\kappa\left(\matA \right)\log\left(\frac{1}{\varepsilon}\right)\right\rceil$, then
$$
\abs{ \alogdet{\matA} - \logdet{\matA} }  \leq 2\epsilon \Gamma.
$$
\end{lemma}

\begin{proof}
First, note that using our choice for $\alpha$ in Step 3 of Algorithm~\ref{alg1a} and applying Lemma~\ref{lem:power}, we get that, with probability at least $1-\delta$,
\begin{equation}\label{eqn:pd11}
\lambda_1(\matA) < \frac 7 6 \lambda_1(\matA) \leq \alpha \le 7  \lambda_1(\matA),
\end{equation}
The strick inequality at the leftmost side of the above equation follows since all eigenvalues of $\matA$ are strictly positive. Let's call the event that the above inequality holds ${\cal E}_1$; obviously, $\ProbB{{\cal E}_1} \geq 1-\delta$ (and thus $\ProbB{\bar{\cal E}_1}\leq \delta$). We condition all further derivations on ${\cal E}_1$ holding and we manipulate $\Delta=\abs{\alogdet{\matA} -
\logdet{\matA}}$ as follows:
\begin{align*}
\Delta&=  \abs{  \sum_{k=1}^{m} \left( \frac{1}{p} \sum_{i=1}^{p}    \g_i^\top \matC^k \g_i  \right)/k - \sum_{k=1}^{\infty} \trace{\matC^k} / k}    \\
&\le
\abs{  \sum_{k=1}^{m}
\left(
\frac{1}{p} \sum_{i=1}^{p}    \g_i^\top \matC^k \g_i
\right)/k -\sum_{k=1}^{m} \trace{\matC^k}/k }
+
\abs{ \sum_{k=m+1}^{\infty} \trace{\matC^k}/k   } \\
&=
\underbrace{\abs{  \frac{1}{p} \sum_{i=1}^{p}
   \g_i^\top \left(\sum_{k=1}^{m} \matC^k/k  \right) \g_i  - \trace{\sum_{k=1}^{m} \matC^k/k }}}_{\Delta_1}
+
\underbrace{\abs{ \sum_{k=m+1}^{\infty} \trace{\matC^k}/k}}_{\Delta_2}.
\end{align*}

\noindent Below, we bound the two terms $\Delta_1$ and $\Delta_2$ separately. We start with $\Delta_1$:
the idea is to apply Lemma~\ref{thm:trace} on the matrix $\sum_{k=1}^{m} \matC^k/k$ with
$p =\left\lceil 20 \log(2/\delta) / \eps^2
\right\rceil$. Let ${\cal E}_2$ denote the probability that Lemma~\ref{thm:trace} holds; obviously, $\ProbB{{\cal E}_2}\geq 1-\delta$ (and thus $\ProbB{\bar{\cal E}_2}\leq \delta$) given our choice of $p$. We condition all further derivations on ${\cal E}_2$ holding as well to get
$$
\Delta_1 \le \eps \cdot \trace{\sum_{k=1}^{m} \matC^k/k} \le \eps \cdot \trace{\sum_{k=1}^{\infty} \matC^k/k}.
$$
In the last inequality we used the fact that $\matC$ is a positive matrix, hence for all $k$, $\trace{\matC^k} > 0$.
The second term $\Delta_2$ is bounded as follows:
\begin{align*}
\Delta_2 &= \abs{\sum_{k=m+1}^{\infty} \trace{\matC^k}/k }
\le \sum_{k=m+1}^{\infty} \trace{\matC^k}/k  \\
&= \sum_{k=m+1}^{\infty} \trace{\matC^m \cdot \matC^{k-m}}/k
\le \sum_{k=m+1}^{\infty} \TNorm{\matC^m} \cdot \trace{ \matC^{k-m}}/k  \\
&= \TNorm{\matC^m} \cdot  \sum_{k=m+1}^{\infty}\trace{ \matC^{k-m}}/k
\le \TNorm{\matC^{m}} \cdot \sum_{k=1}^{\infty} \trace{\matC^k}/k  \\
&\le \left(1-\frac{\lambda_{n}\left(\matA\right)}{\alpha}\right)^m \cdot \sum_{k=1}^{\infty} \trace{\matC^k}/k.
\end{align*}
In the first inequality, we used the triangle inequality and the fact that $\matC$ is a positive matrix.
In the second inequality, we used the following fact\footnote{This follows from Von Neumann's trace inequality.}: given two positive semidefinite matrices
$\matA,\matB$ of the same size,
$
\trace{\matA \matB} \leq \TNorm{\matA} \cdot \trace{\matB}.
$
In the last inequality, we used the fact that
$$
\lambda_{1}(\matC) = 1 - \lambda_{n}(\matB) = 1 - \lambda_{n}(\matA) /\alpha.
$$
Combining the bounds for $\Delta_1$ and $\Delta_2$ gives
\begin{align*}
\abs{ \alogdet{\matA} - \logdet{\matA} } & \leq 
\left(\epsilon+\left(1-\frac{\lambda_{n}\left(\matA\right)}{\alpha}\right)^m\right) \cdot \sum_{k=1}^{\infty} \frac{\trace{\matC^k}}{k}.
\end{align*}
We have already proven in Lemma~\ref{lem1} that
$$ \sum_{k=1}^{\infty} \frac{\trace{\matC^k}}{k} = - \trace{\logm{\matB}} = n \log(a) - \logdet{\matA}.$$
Notice that the assumption of Lemma~\ref{lem1} (namely, $\lambda_1(\matA) < \alpha$) is satisfied from the inequality of eqn.~(\ref{eqn:pd11}). We further manipulate the last term as follows:
\begin{eqnarray*}
n \log(a) - \logdet{\matA}
&=& n \log(\alpha) - \log ( \prod_{i=1}^n \lambda_i(\matA)) \\
&=& n \log(\alpha) - \sum_{i=1}^n \log (\lambda_i(\matA)) \\
&=&  \sum_{i=1}^n \left( \log\left(\alpha\right) - \log\left(\lambda_i(\matA) \right) \right) \\
&=&  \sum_{i=1}^n  \log\left( \frac{\alpha}{ \lambda_i(\matA) }\right).
\end{eqnarray*}
Collecting our results together, we get:
\begin{equation*}
\abs{ \alogdet{\matA} - \logdet{\matA} }
\leq
\left(\epsilon+\left(1-\frac{\lambda_{n}\left(\matA\right)}{\alpha}\right)^m\right) \cdot  \sum_{i=1}^n  \log\left( \frac{\alpha}{ \lambda_i(\matA) }\right) .
\end{equation*}
Using the inequality of eqn.~(\ref{eqn:pd11}) (only the upper bound on $\alpha$ is needed here) proves the first inequality of the lemma. To prove the second inequality, we use the well-known fact that $\left(1-x^{-1}\right)^x \leq e^{-1}$ (where $e=2.718\ldots$ and $x>0$) and our choice for $m$.

Finally, recall that we conditioned all derivations on events ${\cal E}_1$ and ${\cal E}_2$ both holding, which can be bounded as follows:
$$\ProbB{{\cal E}_1\cap{\cal E}_2} = 1-\ProbB{\bar{\cal E}_1\cup\bar{\cal E}_2}\geq 1-\ProbB{\bar{\cal E}_1}-
\ProbB{\bar{\cal E}_1}\geq 1-2\delta.$$
The first inequality in the above derivation follows from the union bound.
\end{proof}

\subsection{Running time} Step 2 takes $\OO(nnz(\matA)\log(n)\log(1/\delta))$ time; we assume that $nnz(\matA)\geq n$, since otherwise the determinant of $\matA$ would be trivially equal to zero.
For each $k>0$, $\v_k = \matC^{k} \g_i$. The algorithm inductively computes $\v_k$ and $\g_i^\top \matC^{k}\g_i = \g_i^\top \v_k$ for all $k=1,2,\ldots, m$. Given $\v_{k-1}$, $\v_{k}$  and $\g_i^\top \matC^{k} \g_i$ can be computed in $nnz(\matC)$ and $\OO(n)$ time, respectively. Notice that $nnz(\matC) \leq n + \nnz(\matA)$. Therefore, step 7 requires
$
\OO( p\cdot m \cdot \nnz(\matA))
$
time. Since $p = O(\eps^{-2}\log(1/\delta)),$ the total cost is
$$
\OO\left( \nnz(\matA) \cdot \left(\frac{m}{\eps^{2}}+\log n\right)\cdot \log\left(\frac{1}{\delta}\right)\right).
$$

\section{Relative error approximation for SPD matrices with bounded eigenvalues}

In this section, we argue that a simplified version of Algorithm~\ref{alg1a} achieves a relative error approximation to the logdet of the SPD matrix $\matA$, under the assumption that all the eigenvalues of $\matA$ lie in the interval $(\theta_1,1)$, where $0< \theta_1 < 1$. This is a mild generalization of the setting introduced in~\cite{icml2015_hana15}.

Given the upper bound on the largest eigenvalue of $\matA$, the proof of the following lemma (which is the analog of Lemma~\ref{lem1}) is straightforward.

\begin{lemma}\label{lem1a}
Let $\matA \in \R^{n \times n}$ be an SPD matrix whose eigenvalues lie in the interval $(\theta_1,1)$, for some $0< \theta_1 < 1$. Let $\matC := \matI_n - \matA$; then,
\begin{equation*}
\logdet{\matA} =  -\sum_{k=1}^{\infty} \frac{\trace{\matC^k}}{k}.
\end{equation*}
\begin{proof}
Similarly to the proof of Lemma~\ref{lem1},
\begin{equation*}
\logdet{\matA}
= \log \left(\prod_{i=1}^{n} \lambda_i(\matA)\right)
= \sum_{i=1}^{n}\log (\lambda_i(\matA))
= \trace{ \logm{\matA}}.
\end{equation*}
\noindent Now,
\begin{equation*}
\trace{ \logm{\matA}}
= \trace{ \logm{\Id_n - (\Id_n - \matA)}}
= \trace{ -\sum_{k=1}^{\infty}  \frac{\left(\Id_n - \matA\right)^k}{k}} = -\sum_{k=1}^{\infty} \frac{\trace{ \matC^k}}{k}.
\end{equation*}
The second equality follows by the Taylor expansion since all the
eigenvalues of $\matC = \Id_n - \matA$ are contained in the interval $(0,1)$.
\end{proof}
\end{lemma}

\subsection{The algorithm and the relative error bound}

We simplify Algorithm~\ref{alg1a} as follows: we skip steps 2 and 3 and in step 4 we set $\matC = \matI_n-\matA$. The following lemma proves that in this special case the modified algorithm returns a relative error approximation to the log determinant of the input matrix $\matA$.
\begin{lemma}\label{thm2}
	Let $\alogdet{\matA}$ be the output of the (modified) Algorithm~\ref{alg1a} on inputs $\matA$ and $\varepsilon$.
	Then, with probability at least $1-\delta$,
	$$
	\abs{ \alogdet{\matA} - \logdet{\matA} }  \leq
	2\varepsilon \cdot \abs{\logdet{\matA}}.$$
\end{lemma}

\begin{proof}
	Similarly to the proof of Lemma~\ref{thm1}, we manipulate $\Delta=\abs{\alogdet{\matA} -
		\logdet{\matA}}$ as follows:
	\begin{align*}
	\Delta&=  \abs{  \sum_{k=1}^{m} \left( \frac{1}{p} \sum_{i=1}^{p}    \g_i^\top \matC^k \g_i  \right)/k - \sum_{k=1}^{\infty} \trace{\matC^k} / k}    \\
	&\le
	\abs{  \sum_{k=1}^{m}
		\left(
		\frac{1}{p} \sum_{i=1}^{p}    \g_i^\top \matC^k \g_i
		\right)/k -\sum_{k=1}^{m} \trace{\matC^k}/k }
	+
	\abs{ \sum_{k=m+1}^{\infty} \trace{\matC^k}/k   } \\
	&=
	\underbrace{\abs{  \frac{1}{p} \sum_{i=1}^{p}
			\g_i^\top \left(\sum_{k=1}^{m} \matC^k/k  \right) \g_i  - \trace{\sum_{k=1}^{m} \matC^k/k }}}_{\Delta_1}
	+
	\underbrace{\abs{ \sum_{k=m+1}^{\infty} \trace{\matC^k}/k}}_{\Delta_2}.
	\end{align*}
	\noindent We now bound the two terms $\Delta_1$ and $\Delta_2$ separately. We start with $\Delta_1$:
	the idea is to apply Lemma~\ref{thm:trace} on the matrix $\sum_{k=1}^{m} \matC^k/k$ with
	$p = \left\lceil 20 \log(2/\delta) / \varepsilon^2 \right \rceil$. Hence, with probability at least $1-\delta$ (this is the only probabilistic event in this lemma and hence $1-\delta$ is a lower bound on the success probability of the lemma):
	$$
	\Delta_1 \le \eps \cdot \trace{\sum_{k=1}^{m} \matC^k/k} \le \eps \cdot \trace{\sum_{k=1}^{\infty} \matC^k/k}.
	$$
	\normalsize
	In the last inequality we used the fact that $\matC$ is a positive definite matrix, hence for all $k$, $\trace{\matC^k} > 0$.
	Bounding $\Delta_2$ follows the lines of the proof of Lemma~\ref{thm1}:
	\begin{align*}
	\Delta_2 &= \abs{\sum_{k=m+1}^{\infty} \trace{\matC^k}/k }
	= \abs{\sum_{k=m+1}^{\infty} \trace{\matC^m \cdot \matC^{k-m}}/k}\\
	&\le \abs{\sum_{k=m+1}^{\infty} \TNorm{\matC^m} \cdot \trace{ \matC^{k-m}}/k}
	= \TNorm{\matC^m} \cdot \abs{\sum_{k=m+1}^{\infty}\trace{ \matC^{k-m}}/k}\\
	&\le \TNorm{\matC^{m}} \cdot \abs{\sum_{k=1}^{\infty} \trace{\matC^k}/k }
	\le \left(1-\lambda_{n}\left(\matA\right)\right)^m \abs{\sum_{k=1}^{\infty} \trace{\matC^k}/k}.
	\end{align*}
	\normalsize
	In the last inequality, we used the fact that $\lambda_{1}(\matC) =  1 - \lambda_{n}(\matA).$
	Combining the bounds for $\Delta_1$ and $\Delta_2$ gives
	\begin{align*}
	\abs{ \alogdet{\matA} - \logdet{\matA} } & \leq 
	\left(\epsilon+\left(1-\lambda_{n}\left(\matA\right)\right)^m\right) \cdot \sum_{k=1}^{\infty} \frac{\trace{\matC^k}}{k}.
	\end{align*}
	\noindent We have already proven in Lemma~\ref{lem1a} that
	$$ \sum_{k=1}^{\infty} \frac{\trace{\matC^k}}{k} = - \trace{\logm{\matA}} =  - \logdet{\matA}.$$
	\noindent Collecting our results, we get:
	\begin{align*}
	\abs{ \alogdet{\matA} - \logdet{\matA} }
	\leq
	\left(\varepsilon+\left(1-\lambda_{n}\left(\matA\right)\right)^m\right) \cdot  \abs{\logdet{\matA}}.
	\end{align*}
	Using
	$
	1-\lambda_n(\matA) < 1-\theta_1,
	$
	we conclude that
		\begin{align*}
		\abs{ \alogdet{\matA} - \logdet{\matA} }
		\leq
		\left(\varepsilon+\left(1-\theta_1\right)^m\right) \cdot  \abs{\logdet{\matA}}.
		\end{align*}
Setting
$$m=\left\lceil\frac{1}{\theta_1}\cdot\log{\left(\frac{1}{\varepsilon}\right)}\right\rceil$$
and using $\left(1-x^{-1}\right)^x \leq e^{-1}$ (where $e=2.718\ldots$ and $x>0$),
guarantees that $(1-\theta_1)^m \leq \varepsilon$ and concludes the proof of the lemma.
\end{proof}

\noindent We conclude by discussing the running time of the simplified Algorithm~\ref{alg1a}, which is equal to $\OO(p\cdot m\cdot nnz(\matA) )$. Since $p = \OO\left(\frac{\log(1/\delta)}{\varepsilon^{2}}\right)$ and $m=\OO\left(\frac{\log(1/\varepsilon)}{\theta_1}\right)$, the running time becomes $$\OO\left(\frac{\log(1/\varepsilon)\log(1/\delta)}{\varepsilon^2\theta_1} nnz(\matA)\right).$$ 

\section{Experiments}\label{sxn:exps}

The goal of our experimental section is to establish that our approximation
to $\logdet{\matA}$ (as computed by Algorithm~\ref{alg1a}) is both accurate and fast for both
dense and sparse matrices.
The accuracy of Algorithm~\ref{alg1a} is measured by comparing its result
against the exact $\logdet{\matA}$ computed via the Cholesky factorization.
The rest of this section is organized as follows:
in Section~\ref{subsec:software}, we describe our software for approximating
$\logdet{\matA}$; in Section~\ref{subsec:environment} we describe the computational environment that we used; and
in Sections~\ref{subsec:dense_matrices} and~\ref{subsec:sparse_matrices} we
discuss experimental results for dense and sparse SPD matrices, respectively.

\subsection{Software}
\label{subsec:software}
We developed high-quality, shared- and distributed-memory parallel \Cpp{}
code for the algorithms listed in this paper.
All of the code that was developed for this paper is hosted at \url{https://github.com/pkambadu/ApproxLogDet}.
In it's current state, our software supports:
(1) ingesting dense (binary and text format) and sparse (binary,
text, and matrix market format) matrices,
(2) generating large random SPD matrices,
(3) computing both approximate and exact spectral norms of matrices,
(4) computing both approximate and exact traces of matrices, and
(5) computing both approximate and exact log determinants of matrices.
Currently, we support both Eigen~\cite{eigenweb} and
Elemental~\cite{poulson2013elemental} matrices.
The Eigen software package supports both dense and sparse matrices, while the Elemental software package mostly
supports dense matrices and only recently added support for sparse matrices
(pre-release).
As we wanted the random SPD generation to be fast, we have used parallel
random number generators from Random123~\cite{salmon2011parallel} in
conjunction with Boost.Random.

\subsection{Environment}
\label{subsec:environment}
All our experiments were run on ``Nadal'', a 60-core machine, where each
core is an Intel\textregistered{} Xeon\textregistered{} E7-4890 machine running
at 2.8 Ghz.
Nadal has 1 TB of RAM and runs Linux kernel version 2.6-32.
For compilation, we used GCC 4.9.2.
We used Eigen 3.2.4, OpenMPI 1.8.4, Boost 1.55.7, and the latest version of
Elemental at {\small\url{https://github.com/elemental}}.
For experiments with Elemental, we used OpenBlas, which is an extension of
GotoBlas~\cite{goto2008high}, for its parallel prowess; Eigen has built-in the
BLAS and LAPACK packages.

\subsection{Dense Matrices}
\label{subsec:dense_matrices}
\vspace{0.01in}\noindent \textbf{Data Generation.}
In our experiments, we used two types of synthetic SPD matrices. The first
type were diagonally dominant SPD matrices and were generated as follows.
First, we created $\matX\in{}\mathbb{R}^{n\times{}n}$ by drawing $n^2$ entries from a
uniform sphere with center 0.5 and radius 0.25.
Then, we generated a symmetric matrix $\matY$ by setting
$$\matY=0.5*(\matX+\matX^\top).$$
Finally, we ensured that the desired matrix $\matA$ is positive definite by adding the value $n$ to
each diagonal entry~\cite{curran2009variation} of $\matY$:
$
\matA = \matY + n \matI_n.
$
We call this method~\code{randSPDDenseDD}.

The second approach generates SPD matrices that are not diagonally dominant.
We created $\matX,\matD \in{} \mathbb{R}^{n\times{}n}$ by drawing $n^2$ and $n$
entries, respectively, from a uniform sphere with center 0.5 and radius 0.25;
$\matD$ is a diagonal matrix with small entries.
Next, we generated an orthogonal random matrix $\matQ = \qr{\matX}$. Thus,
$\matQ$ is an orthonormal basis for $\matX$.
Finally, we generated
$
\matA = \matQ\matD{}\matQ^T.
$
We call this method \code{randSPDDense}.
\code{randSPDDense} is more expensive than \code{randSPDDenseDD}, as it requires
an additional $O(n^3)$ computations for the QR factorization and the
matrix-matrix product.

\vspace{0.02in}\noindent \textbf{Evaluation.}
To evaluate the runtime of Algorithm~\ref{alg1a} against a baseline, we used the Cholesky decomposition
to compute the $\logdet{\matA}$. More specifically, we computed $\matA = \matL\matL^T$ and
returned $\logdet{\matA} = 2\cdot\logdet{\matL}$.
Since Elemental provides distributed and shared memory parallelism, we restricted
ourselves to experiments with Elemental matrices throughout this section.
Note that we measured the accuracy of the approximate algorithm in terms of the
relative error to ensure that we have numbers of the same scale for matrices
with vastly different values for $\logdet{\matA}$;
we defined the relative error $e$ as $e =
100(x-\tilde{x})/x$, where $x$ is the true value and
$\tilde{x}$ is the approximation.
Similarly, we defined the speedup $s$ as $s = t_x/t_{\tilde{x}}$,
where $t_x$ is the time needed to compute $x$ and $t_{\tilde{x}}$ is the time needed to
compute the approximation $\tilde{x}$.

\vspace{0.02in}\noindent \textbf{Results.}
For dense matrices, we first used synthetic matrices generated using~\code{randSPDDense}; these are relatively ill-conditioned matrices. We experimented with values of $n$ (number of rows and columns of $\matA$) in the set $\{5,000,\ 7,500,\ 10,000,\ 12,500,\ 15,000\}$.
The three key points pertaining to these matrices are shown in
Figure~\ref{fig:dense}.
First, we discuss the effect of $m$, the number of terms in the Taylor series
used to approximate $\logdet{\matA}$; Figure~\ref{fig:dense-accuracy} depicts
our results for the sequential case.
On the $y$-axis, we see the relative error, which is measured against the exact
$\logdet{\matA}$ as computed via the Cholesky factorization.
We observe that for these ill-conditioned matrices, for small values of $m$ (less than four)
the relative error is high.
However, for all values of $m\geq 4$, we observe that the error drops
significantly and stabilizes.
We note that in each iteration, all random processes were re-seeded with new
values; we have plotted the error bars throughout Figure~\ref{fig:dense}.
The standard deviation for both accuracy and time was consistently
small; indeed, it is not visible to the naked eye at scale.
To see the benefit of approximation, we look at
Figure~\ref{fig:dense-speedup-m} together with Figure~\ref{fig:dense-accuracy}.
For example, at $m=4$, for all matrices, we get at least a factor of two speedup.
As $n$ gets larger, the speedups of the
approximation also increase.
For example, for $n=15,000$, the speedup at $m=4$ is nearly six-fold. In terms of accuracy,
Figure~\ref{fig:dense-accuracy} shows that at $m=4$, the relative error is approximately
$4\%$.
This speedup is expected as the Cholesky factorization requires $O(n^3)$
operations; Algorithm~\ref{alg1a} only relies on matrix-matrix products where one
of the matrices has a small number of columns (equal to $p$), which is
independent of $n$.
\begin{figure*}[t]
\begin{center}
\subfigure[Accuracy vs. $m$]{%
\includegraphics[width=0.32\textwidth]{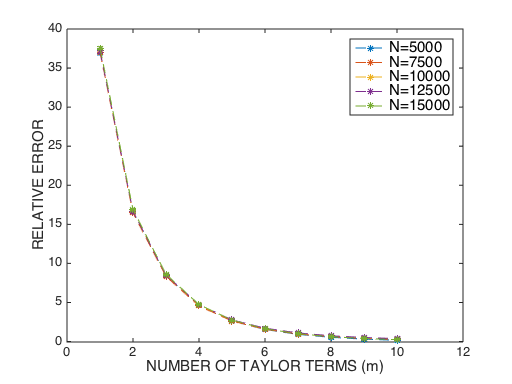}
\label{fig:dense-accuracy}
}
\subfigure[Speedup vs. $m$]{%
\includegraphics[width=0.32\textwidth]{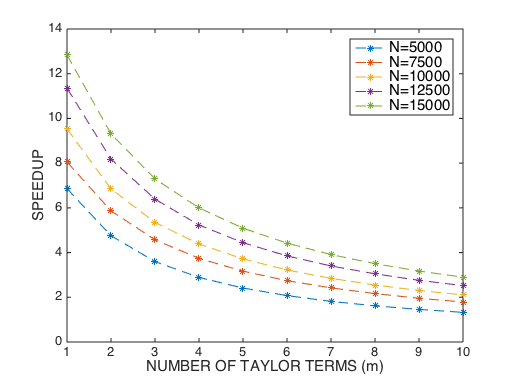}
\label{fig:dense-speedup-m}
}
\subfigure[Parallel Speedup]{%
\includegraphics[width=0.32\textwidth]{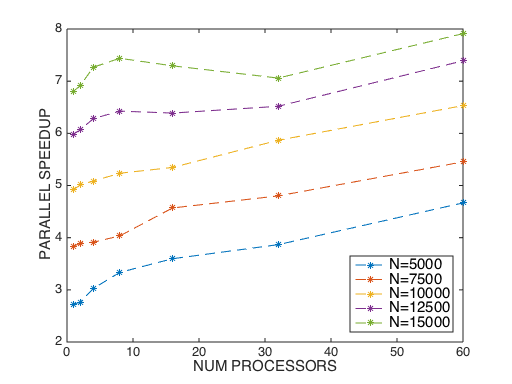}
\label{fig:dense-speedup-np}
}
\end{center}
\caption{
Panels~\ref{fig:dense-accuracy} and~\ref{fig:dense-speedup-m} depict the effect
of $m$ (see Algorithm~\ref{alg1a}) on the accuracy of the approximation and the
time to completion, respectively, for dense matrices generated by
\code{randSPDDense}.
For all the panels, $p=60$ and $t=2\log\sqrt{4n}$.
The baseline for all experiments was the Cholesky factorization, which was used
to compute the exact value of $\logdet{\matA}$.
For panels~\ref{fig:dense-accuracy} and~\ref{fig:dense-speedup-m}, the number
of cores, $np$, was set to one.
The last panel~\ref{fig:dense-speedup-np} depicts the relative speedup of
the approximate algorithm when compared to the baseline solver (at $m=4$).
Elemental was used as the backend for these experiments.
For the approximate algorithm, we report the mean and standard deviation
of ten iterations.
}
\label{fig:dense}
\end{figure*}
\begin{figure*}[t]
\begin{center}
\subfigure[Accuracy vs. $m$]{%
\includegraphics[width=0.32\textwidth]{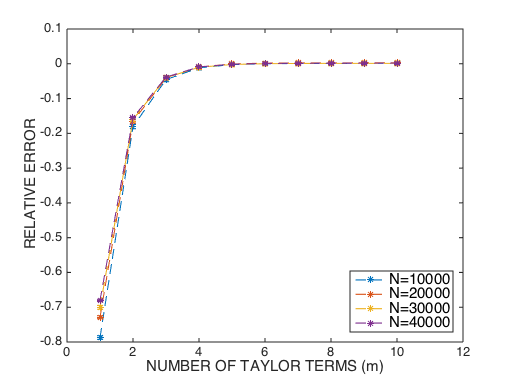}
\label{fig:dense-dd-accuracy}
}
\subfigure[Speedup vs. $m$]{%
\includegraphics[width=0.32\textwidth]{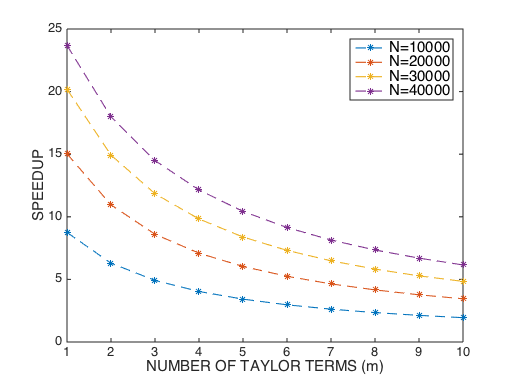}
\label{fig:dense-dd-speedup-m}
}
\subfigure[Parallel Speedup]{%
\includegraphics[width=0.32\textwidth]{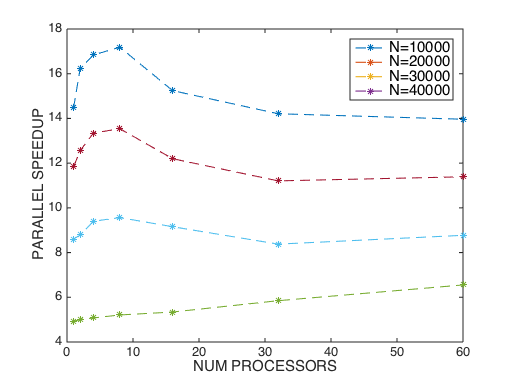}
\label{fig:dense-dd-speedup-np}
}
\end{center}
\caption{
Panels~\ref{fig:dense-dd-accuracy} and~\ref{fig:dense-dd-speedup-m} depict the
effect of $m$ (see Algorithm~\ref{alg1a}) on the accuracy of the approximation
and the time to completion, respectively, for diagonally dominant dense random
matrices generated by \code{randSPDDenseDD}.
For all the panels, $p=60$ and $t=2\log\sqrt{4n}$.
The baseline for all experiments was the Cholesky factorization, which was used
to compute the exact value of $\logdet{\matA}$.
For panels~\ref{fig:dense-dd-accuracy} and~\ref{fig:dense-dd-speedup-m}, the
number of cores, $np$, was set to one.
The last panel~\ref{fig:dense-dd-speedup-np} depicts the relative speedup of
the approximate algorithm when compared to the baseline solver (at $m=2$).
Elemental was used as the backend for these experiments.
For the approximate algorithm, we report the mean and standard deviation
over ten iterations.
}
\label{fig:dense-dd}
\end{figure*}
\begin{table}
\center
\scriptsize
\begin{tabular}{|c|c|c|c|c|c|c|}
\hline
\multirow{2}{*}{$n$}  &
\multicolumn{3}{|c|}{$\logdet{\matA}$} &
\multicolumn{3}{|c|}{time (secs)} \\
\cline{2-7}
      &   exact    & mean       & std   &   exact  & mean  &  std   \\\hline
5000  & -3717.89  & -3546.920  &  8.10  &   2.56   &  1.15 & 0.0005 \\\hline
7500  & -5474.49  & -5225.152  &  8.73  &   7.98   &  2.53 & 0.0015 \\\hline
10000 & -7347.33  & -7003.086  &  7.79  &  18.07   &  4.47 & 0.0006 \\\hline
12500 & -9167.47  & -8734.956  & 17.43  &  34.39   &  7.00 & 0.0030 \\\hline
15000 & -11100.9  & -10575.16  & 15.09  &  58.28   & 10.39 & 0.0102 \\\hline
\end{tabular}
\caption{
Accuracy and sequential running times (at $p=60$, $m=4$ and $t=\log\sqrt{4n}$) for dense
random matrices generated using \code{randSPDDense}.
Baselines were computed using the Cholesky factorization; mean and standard
deviation are reported over ten iterations.
}
\label{tbl:dense-abs}
\normalsize
\end{table}
\begin{table}
\center
\scriptsize
\begin{tabular}{|c|c|c|c|c|c|c|}
\hline
\multirow{2}{*}{$n$}  &
\multicolumn{3}{|c|}{$\logdet{\matA}$} &
\multicolumn{3}{|c|}{time (secs)} \\
\cline{2-7}
      &   exact   & mean     & std    &   exact  & mean  &  std   \\\hline
10000 &  92103.1  &  92269.5 &  5.51  &   18.09  &  2.87 &  0.01  \\\hline
20000 & 198069.0  & 198397.4 &  9.60  &  135.92  & 12.41 &  0.02  \\\hline
30000 & 309268.0  & 309763.8 & 20.04  &  448.02  & 30.00 &  0.12  \\\hline
40000 & 423865.0  & 424522.4 & 14.80  & 1043.74  & 58.05 &  0.05  \\\hline
\end{tabular}
\caption{
Accuracy and sequential running times (at $p=60$, $m=2$, and $t=2\log\sqrt{4n}$) for
diagonally dominant dense random matrices generated using
\code{randSPDDenseDD}.
Baselines were computed using the Cholesky factorization; mean and standard
deviation are reported over ten iterations.
}
\label{tbl:dense-dd-abs}
\normalsize
\end{table}
%
%
Finally, we discuss the parallel speedup in Figure~\ref{fig:dense-speedup-np},
which shows the relative speedup of the approximate algorithm with respect to
the baseline Cholesky algorithm.
For this evaluation, we set $m=4$ and varied the number of processes, denoted by $np$,
from $1$ to $60$.
The main take away from Figure~\ref{fig:dense-speedup-np} is that the
approximate algorithm provides nearly the same or increasingly better speedups
relative to a parallelized version of the exact (Cholesky) algorithm.
For example, for $n=15,000$, the speedups for using the approximate algorithm
are consistently better that $6.5x$.
The absolute values for $\logdet{\matA}$ and timing along with the baseline
numbers for this experiment are given in Table~\ref{tbl:dense-abs}.
We report the numbers in Table~\ref{tbl:dense-abs} at $m=4$ at which point, we
have low relative error.

For the second set of dense experiments, we generated diagonally dominant
matrices using \code{randSPDDenseDD}; we were able to quickly generate and
run benchmarks on matrices of sizes $n \times n$ with $n$ in the set $\{10,000,\ 20,000,\ 30,000,\ 40,000\}$ due to the
relatively simpler procedure involved in matrix generation.
In this set of of experiments, due to the diagonal dominance, all matrices
were well-conditioned.
The results of our experiments on these well-conditioned matrices are
presented in Figure~\ref{fig:dense-dd} and show a marked improved over the
results presented in Figure~\ref{fig:dense}.
First, notice that very few terms of the Taylor series (i.e., small $m$) are
sufficient to get high accuracy approximations; this is apparent in Figure~\ref{fig:dense-dd-accuracy}.
In fact, we see that even at $m=2$, we are near convergence and at $m=3$, for
most of the matrices, we have near-zero relative error.
This experimental result, combined with Figure~\ref{fig:dense-dd-speedup-m} is particularly encouraging; at $m=2$, we seem to not only have a nearly lossless approximation of $\logdet{\matA}$, but also have at least a five-fold speedup.
Similarly to Figure~\ref{fig:dense}, the speedups are better for larger matrices.
For example, for $n=40,000$, the speedup at $m=2$ is nearly twenty-fold.
We conclude our analysis by presenting Figure~\ref{fig:dense-dd-speedup-np},
which similarly to Figure~\ref{fig:dense-speedup-np}, points out that at any
level of parallelism, Algorithm~\ref{alg1a} maintains its relative performance
over the exact (Cholesky) factorization.
The absolute values for $\logdet{\matA}$ and the corresponding running times, along with the baseline
for this experiment are presented in Table~\ref{tbl:dense-dd-abs}.
We report the numbers in Table~\ref{tbl:dense-abs} at $m=2$, at which point we
have a low relative error.
%

\subsection{Sparse Matrices}

\label{subsec:sparse_matrices}
\vspace{0.02in}\noindent \textbf{Data Synthesis.}
To generate a sparse, synthetic matrix $\matA\in{}\mathbb{R}^{n\times{}n}$, with
$nnz$ non-zeros, we use a Bernoulli distribution to determine the location of
the non-zero entries and a uniform distribution to generate the values.
First, we completely fill the $n$ principle diagonal entries.
Next, we generate $(nnz-n)/2$ index positions in the upper triangle for
the non-zero entries by sampling from a Bernoulli distribution with
probability $(nnz-n)/(n^2-n)$.
We reflect each entry across the principle diagonal to ensure that $\matA$ is symmetric and we add $n$ to each diagonal entry to ensure that $\matA$ is SPD (actually, $\matA$ is also diagonally dominant).

\vspace{0.02in}\noindent \textbf{Real Data.}
\noindent
To demonstrate the prowess of Algorithm~\ref{alg1a} on real-world data, we
used SPD matrices from the University of Florida's sparse matrix
collection~\cite{davis2011university}.
The complete list of matrices from this collection used in our experiments, as well as a brief description of each matrix, is given in columns 1--4 of Table~\ref{tbl:ufl}.
\begin{table*}
\scriptsize
\centering
\begin{tabular}{|c|c|c|c|c|c|c|c|c|c|}
\hline
\multirow{3}{*}{name} &
\multirow{3}{*}{$n$} &
\multirow{3}{*}{$nnz$} &
\multirow{3}{*}{area of origin} &
\multicolumn{3}{|c|}{\logdet{\matA}} &
\multicolumn{2}{|c|}{time (sec)} &
\multirow{3}{*}{$m$} \\ \cline{5-9}

 & & & & \multirow{2}{*}{exact}  &
         \multicolumn{2}{|c|}{approx} &
         \multirow{2}{*}{exact} &
         approx & \\ \cline{6-7} \cline{9-9}

 & & & & & mean & std & & mean & \\\hline
thermal2        & 1228045  & 8580313   & Thermal         & 1.3869e6   & 1.3928e6  & 964.79 & 31.28  & 31.24 & 149 \\ \hline
ecology2        & 999999   & 4995991   & 2D/3D           & 3.3943e6   & 3.403e6   & 1212.8 & 18.5   & 10.47 & 125 \\ \hline
ldoor           & 952203   & 42493817  & Structural      & 1.4429e7   & 1.4445e7  & 1683.5 & 117.91 & 17.60 &  33 \\ \hline
thermomech\_TC  & 102158   & 711558    & Thermal         & -546787    & -546829.4 & 553.12 & 57.84  &  2.58 &  77 \\ \hline
boneS01         & 127224   & 5516602   & Model reduction & 1.1093e6   & 1.106e6   & 247.14 & 130.4  &  8.48 & 125 \\ \hline
\end{tabular}
\caption{
Description of the SPD matrices from the University of Florida sparse
matrix collection~\cite{davis2011university} that were used in our experiments.
All experiments were run sequentially ($np=1$) using Eigen.
Accuracy results for Algorithm~\ref{alg1a} are reported using both the mean and the standard
deviation over ten iterations at (with $t=5$ and $p=5$); we only report the mean for the running times, since the standard deviation is negligible.
The exact $\logdet{\matA}$ was computed using the Cholesky factorization.
}
\label{tbl:ufl}
\end{table*}

\vspace{0.02in}\noindent \textbf{Evaluation.}
\noindent
It is tricky to pick any single method as the ``exact method'' to compute the
$\logdet{\matA}$ for a sparse SPD matrix $\matA$.
One approach would be to use direct methods such as Cholesky decomposition of
$\matA$~\cite{davis2006direct,gupta2000wsmp}.
For direct methods, it is difficult to derive an analytical solution
for the number of operations required for the factorization as a function of the
number of non-zero entries of the matrix, as this is highly dependent on the structure of
the matrix~\cite{gupta1997highly}.
In the distributed setting, one also needs to consider the volume of
communication involved, which is often the bottleneck.
Alternately, we can use iterative methods to compute the eigenvalues of
$\matA$~\cite{davidson1975iterative} and use the eigenvalues to compute
$\logdet{\matA}$.
It is clear that the worst case performance of both the direct and iterative
methods is $O(n^3)$.
However, iterative methods are typically used to compute a few eigenvalues and
eigenvectors: therefore, we chose to use the Cholesky factorization based on
matrix reordering to compute the exact value of $\logdet{\matA}$.
It is important to note that both the direct and iterative methods are
notoriously hard to implement, which comes to stark contrast with the almost trivial implementation of Algorithm~\ref{alg1a}, which is also readily parallelizable.
\begin{figure}[t]
\begin{center}
\subfigure[Convergence with $m$]{%
\includegraphics[height=0.25\textwidth,width=0.45\textwidth]{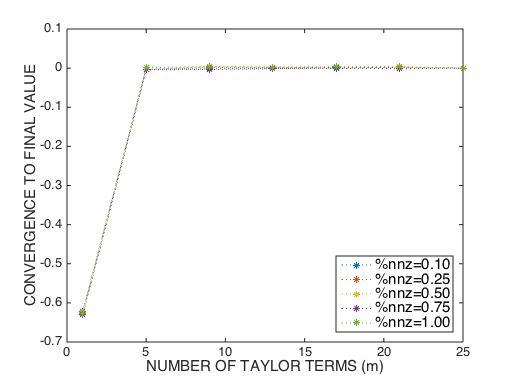}
\label{fig:sparse-convergence-m}
}
\subfigure[Cost as a function of $m$]{%
\includegraphics[height=0.25\textwidth,width=0.45\textwidth]{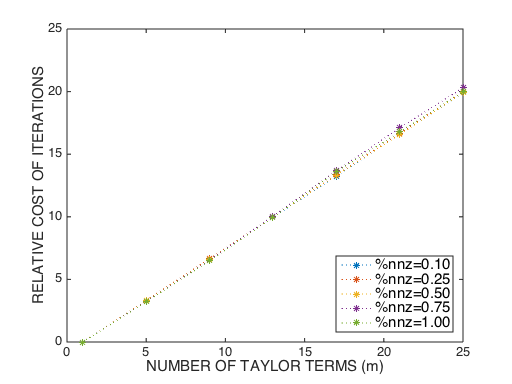}
\label{fig:sparse-cost-m}
}
\end{center}
\caption{
Panels~\ref{fig:sparse-convergence-m} and~\ref{fig:sparse-cost-m} depict the
effect of the number of terms in the Taylor expansion, $m$, (see
Algorithm~\ref{alg1a}) on the convergence to the final solution and the time to
completion of the approximation.
The matrix size was fixed at $n=10^6$ and sparsity was varied as $0.1\%,
0.25\%, 0.5\%, 0.75\%, \textnormal{ and } 1\%$.
Experiments were run sequentially ($np=1$) and we set $p=60$, $t=2\log\sqrt{4n}$.
For panel~\ref{fig:sparse-convergence-m}, the baseline is the final value of
$\logdet{\matA}$ at $m=25$.
For panel~\ref{fig:sparse-cost-m}, the baseline is the time to completion of the approximation algorithm with $m=1$.
Eigen was used as the backend for these experiments.
}
\label{fig:sparse}
\end{figure}

\vspace{0.02in}\noindent \textbf{Results.}
The true power of Algorithm~\ref{alg1a} lies in its ability to approximate
$\logdet{\matA}$ for sparse $\matA$.
The Cholesky factorization can introduce $O(n^2)$ non-zeros during
factorization due to fill-in; for many problems, there is insufficient
memory to factorize a large, sparse matrix.
In our first set of experiments, we wanted to show the effect of $m$ on: (1)
convergence of $\logdet{\matA}$, and (2) cost of the solution.
To this end, we generated sparse, diagonally dominant SPD matrices of size
$n=10^6$ and varied the sparsity from $0.1\%$ to $1\%$ in increments of
$0.25\%$.
We did not attempt to compute the exact $\logdet{\matA}$ for these synthetic
matrices --- our aim was to merely study the speedup with $m$ for different
sparsities, while $t$ and $p$ were held constant at $2\log\sqrt{4n}$ and $60$
respectively.
The results are shown in Figure~\ref{fig:sparse}.
Figure~\ref{fig:sparse-convergence-m} depicts the convergence of
$\logdet{\matA}$ measured as a relative error of the current estimate over the
final estimate.
As can be seen --- for well conditioned matrices --- convergence is quick.
Figure~\ref{fig:sparse-cost-m} shows the relative cost of increasing $m$; here
the baseline is $m=1$.
Therefore, the additional cost incurred by increasing $m$ is linear when
all other parameters are held constant.

%
The results of running Algorithm~\ref{alg1a} on the UFL matrices are shown in
Table~\ref{tbl:ufl}.
%
The numbers reported for the approximation are the mean and standard deviation
over ten iterations, $t=5$, and $p=5$~\footnote{We experimented with different
$p,t$ and settled on the smallest values that did not result in loss in
accuracy.}.
The value of $m$ was varied between one and 150 in increments of five to select the best
average accuracy.
The matrices shown in Table~\ref{tbl:ufl} have a nice structure, which lends
itself to nice reorderings and therefore an efficient computation of the Cholesky factorization.
We see that even in such cases, the performance of Algorithm~\ref{alg1a} is
commendable due to its lower
algorithmic complexity; \texttt{ldoor} is the only exception as the
approximation takes longer to compute than the Cholesky factorization.
In the case of \texttt{thermomech\_TC}, we achieve good accuracy while achieving
a 22x speedup.
%
%

\section{Conclusions}\label{sec:conclusions}

Prior work has presented approximation algorithms for the logarithm of the determinant of
a symmetric positive definite matrix; those algorithms either do not work for all SPD matrices, or do
not admit a worst-case theoretical analysis, or both.
In this work, we presented an approximation algorithm to compute the logarithm
of the determinant of a SPD matrix that comes with strong theoretical worst-case
analysis bounds and can be applied to \emph{any} SPD matrix. A simplification of our algorithm delivers relative-error approximation guarantees for a popular special case of SPD matrices.
Using state-of-the-art \Cpp{} numerical linear algebra software packages for
both dense and sparse matrices, we demonstrated that the proposed approximation
algorithm performs remarkably well in practice in serial and parallel
environments.

\newpage
\bibliographystyle{alpha}
\bibliography{references}

\end{document}